\ifdefined\extendedversion\else\newcommand{\extendedversion}\fi

\documentclass[runningheads]{llncs}
\usepackage[T1]{fontenc}
\usepackage{graphicx}
\usepackage{color}
\usepackage{amssymb,amsmath}
\usepackage{mathtools} 
\usepackage{xspace}
\usepackage{xcolor}
\usepackage{multirow}
\usepackage{listings}
\usepackage{proof}
\usepackage{placeins}
\usepackage{svg}

\renewcommand{\vec}[1]{\mathaccent"017E{#1}}

\newcommand{\Fig}{Fig.}

\newcommand{\scythe}{\textsc{Scythe}\xspace}
\newcommand{\tool}{\textsc{PolySemist}\xspace}

\newcommand{\True}{\textit{True}\xspace}
\newcommand{\False}{\textit{False}\xspace}

\newcommand{\Acts}{\textit{Acts}\xspace}

\DeclarePairedDelimiter{\angles}{\langle}{\rangle}

\DeclarePairedDelimiter{\sem}{[\![}{]\!]}

\newcommand{\TLA}{TLA\textsuperscript{+}\xspace}

\newcommand{\Lang}[1]{\mathcal{L}(#1)}

\newcommand{\Vars}{\textsc{Vars}\xspace}
\newcommand{\Init}{\textsc{Init}\xspace}
\newcommand{\Next}{\textsc{Next}\xspace}
\newcommand{\Holes}{\textsc{Holes}\xspace}

\newcommand{\prop}{\ensuremath{\Phi}\xspace}

\usepackage{soul,xcolor}
\setstcolor{red}

\begin{document}

\title{Accelerating Protocol Synthesis and Detecting Unrealizability with Interpretation Reduction}
\titlerunning{Accelerating Protocol Synthesis and Detecting Unrealizability}
%
\author{Derek Egolf
\and Stavros Tripakis
}
\authorrunning{D. Egolf \& S. Tripakis}
%
\institute{Northeastern University, Boston MA, USA\\
\email{\{egolf.d,stavros\}@northeastern.edu}}
\maketitle              
\begin{abstract}
We present a novel counterexample-guided, sketch-based method for the synthesis
of symbolic distributed protocols in \TLA. Our method's chief novelty lies in a
new search space reduction technique called interpretation reduction, which
allows to not only eliminate incorrect candidate protocols before they are sent
to the verifier, but also to avoid enumerating redundant candidates in the
first place. Further performance improvements are achieved by an advanced
technique for exact generalization of counterexamples. Experiments on a set of
established benchmarks show that our tool is almost always faster than the
state of the art, often by orders of magnitude, and was also able to synthesize
an entire \TLA protocol ``from scratch'' in less than 3 minutes where the state
of the art timed out after an hour. Our method is sound, complete, and
guaranteed to terminate on unrealizable synthesis instances under common
assumptions which hold in all our benchmarks.
\end{abstract}

\section{Introduction}
\label{sec:intro}

Distributed protocols are the foundation of many modern computer systems. They
find use in domains such as finance~\cite{Buterin2013,Buchman2016TendermintBF}
and cloud computing~\cite{corbett2013spanner,decandia2007dynamo}. Distributed
protocols are notoriously difficult to design and verify. These difficulties
have motivated the advances in fuller automation for verification of both
safety
\cite{2021ic3posymmetry,2021swisshance,yao2021distai,YaoTGN22,DBLP:conf/fmcad/SchultzDT22,schultz2024scalablearxiv}
and liveness~\cite{YaoLivenessPOPL2024} properties of distributed protocols. 

In this paper, we study the automated {\em synthesis} of distributed
protocols. In verification, the goal is to check whether a given protocol
satisfies a given property. The goal in synthesis is to {\em invent} a protocol
that satisfies a given property. Synthesis is a challenging problem
\cite{PnueliRosner89,DBLP:journals/corr/JacobsB14,PnueliRosner90,Thistle2005,TripakisIPL}, and we use
the  {\em sketching} paradigm~\cite{LezamaAPLAS2009,ArmandoSTTT2013}  to turn
the synthesis problem into a search problem. The user provides a protocol {\em sketch},
which is a protocol with holes in it; each hole is associated
with a grammar. The goal is to choose expressions from the grammars to fill the
holes in the sketch to obtain a protocol that satisfies a given property. Our
target protocols are represented symbolically in the \TLA
language~\cite{lamport2002specifying}, a popular specification language for
distributed systems used in both academia and industry~\cite{NewcombeAmazon2015}.

In this paper we propose a novel method for synthesis of distributed protocols by sketching.
Our method follows the {\em counterexample-guided inductive synthesis} (CEGIS) paradigm~\cite{ArmandoSTTT2013,GulwaniPolozovSingh2017}, but with some crucial innovations, explained below.
In CEGIS, a 
 {\em learner} is responsible for
generating candidate protocols, and a {\em verifier} is responsible for checking
whether the candidate protocols satisfy the property. In addition to saying
correct/incorrect, the verifier also provides a counterexample when the
candidate protocol is incorrect. 
Existing approaches {\em generalize} the counterexample to a set of constraints
which allows to {\em prune} the search
space~\cite{ScenariosHVC2014,DBLP:conf/cav/AlurRSTU15,egolf2024-arxiv}.
Such pruning eliminates candidate protocols that are sure to
exhibit previously encountered counterexamples before they make it to the
verifier, which is crucial for scalability: e.g.,~\cite{egolf2024-arxiv} show that pruning constraints can be used
to prevent hundreds of thousands of model-checker calls.


Unfortunately, pruning alone only goes so far. Even when many calls to the
verifier are avoided, the sheer size of the search space is often huge.
Therefore, enumerating all candidates (and checking each one of them against
the pruning constraints) can be prohibitive in itself. 

In this paper, we address this problem by introducing a novel 
search space reduction technique called {\em interpretation reduction}. While
pruning constraints eliminate incorrect candidate protocols before they are
sent to the verifier, reduction avoids enumerating redundant candidates in the
first place. Reduction can be achieved by choosing an equivalence relation,
partitioning the search space into classes of equivalent expressions, and
enumerating at most one expression from each equivalence class. One choice of
equivalence relation is {\em universal equivalence}. Two expressions are
universally equivalent if they evaluate to the same value under all
interpretations (an {\em interpretation} is an assignment mapping terminal
symbols of sketch grammars to values). E.g., $x+y$ and $y+x$ are universally equivalent.
Our expression enumerator does better, and will only generate semantically
distinct expressions up to a coarser notion of equivalence, namely, {\em
interpretation equivalence}. If $\alpha$ is an interpretation, then two
expressions, $e_1$ and $e_2$, are interpretation equivalent with respect to
$\alpha$ if they evaluate to the same value under $\alpha$. E.g., $x$ and $x+x$
are not universally equivalent, but
are interpretation equivalent with respect to the interpretation $\alpha =
[x\mapsto 0]$. 
Our experiments show that our method is almost always faster than the state of the art,
and in many cases, 
more than 100 times faster. 

In addition to improving efficiency, interpretation reduction also enables
recognizing {\em unrealizable} synthesis instances, i.e., instances where there is no solution.
Under
certain conditions (met in all of our benchmarks), interpretation
reduction partitions the search space into finitely many equivalence classes.
Then, if the algorithm has enumerated an expression from each equivalence class
and none of them satisfy the property, the algorithm can infer that the
synthesis instance is unrealizable and terminate.
Our experiments show that our method
is able to
recognize five times more unrealizable synthesis instances than 
the state of the art.

In summary, this work makes the following contributions: 
(1) a novel search space reduction technique based on interpretation equivalence, 
(2) the first, to our knowledge, method for synthesizing \TLA protocols by
sketching that is guaranteed to terminate on unrealizable synthesis instances
when there are a finite number of equivalence classes modulo interpretation
equivalence,
and (3) a synthesis tool called \tool that implements our method and
outperforms the state of the art. Most notably, \tool was able to synthesize an
entire \TLA protocol from scratch in less than 3 minutes where the state of the
art timed out after an hour. To the best of our knowledge, ours is the first
tool to synthesize a \TLA protocol from scratch.

\section{Distributed Protocol Synthesis}
\label{sec:prelims}


\subsubsection{Protocol Representation in \TLA}

We consider symbolic transition systems modeled in
\TLA~\cite{lamport2002specifying}, e.g., as shown in
\Fig~\ref{fig:protocol_example}. A primed variable, e.g.,
$\textit{vote\_yes}'$, denotes the value of the variable in the next state.
Formally, a {\em protocol} is a tuple $\angles{\Vars, \Init, \Next}$. \Vars is
the set of {\em state variables} (e.g. \Fig~\ref{fig:protocol_example}
line~\ref{line:vars}). \Init and \Next are predicates specifying, respectively,
the {\em initial states} and the {\em transition relation} of the system,  as
explained in detail in the next subsections. 

In \Fig~\ref{fig:protocol_example}, the protocol is parameterized by the set of
nodes participating in the protocol, which is denoted by the declaration on
line \ref{line:constant}. As in
\cite{egolf2024-arxiv}, we are able to synthesize parameterized protocols. For
space reasons, we omit all discussion of parameterized protocols in this paper
and refer the reader to \cite{egolf2024-arxiv} for details.

\subsubsection{Protocol Semantics}

A {\em state} of a protocol is an assignment of values to the variables
in \Vars. 
We write $s[v]$ to denote the value of the state variable $v$ in state $s$.
\Init is a predicate mapping a state to true or false; if a state
satisfies \Init (if it maps to true), it is an initial state of the protocol.
\ifdefined\extendedversion
For instance, the \Init predicate beginning on line~\ref{line:init} of
\Fig~\ref{fig:protocol_example} requires that all shown state variables are
initially the empty set. 
\fi

The transition relation \Next is a predicate mapping a pair of states to true or
false. If a pair of states $(s,t)$ satisfies \Next, then there is a {\em
transition} from $s$ to $t$, and we write $s \rightarrow t$. A state is {\em
reachable} if there exists a {\em run} of the protocol containing that
state. A run of a protocol is a possibly infinite sequence of states
$s_0, s_1, s_2 ...$ such that (1) $s_0$ satifies \Init, (2) $s_i \rightarrow
s_{i+1}$ for all $i \geq 0$, and (3) the sequence satisfies optional {\em
fairness constraints}. We omit a detailed discussion of fairness, but
it is used to exclude runs where a transition is never taken, even though it
was enabled infinitely often.

\subsubsection{Properties and Verification}

We support standard temporal safety and liveness {\em properties} for specifying
protocol correctness. Safety is often specified using a state {\em invariant}: a
predicate mapping a state to true or false. A protocol satisfies a
state invariant if all reachable states satisfy the invariant. A protocol
satisfies a temporal property if all runs (or fair runs, if fairness is
assumed) satisfy the property. 

\begin{figure}
\centering
\begin{minipage}[b]{0.48\textwidth}
\begin{lstlisting}[mathescape,numbers=left,xleftmargin=20pt,escapechar=|]
$\text{CONSTANT Node}$|\label{line:constant}|
$\textit{vars} := (\textit{vote\_yes},\textit{go\_commit},\textit{go\_abort})$|\label{line:vars}|
$\textit{GoCommit} :=$|\label{line:go-commit}|
    $\wedge\ \textit{vote\_yes} = \textit{Node}$|\label{line:go-commit-vote-yes}|
    $\wedge\ \textit{go\_commit}' = \textit{Node}$|\label{line:go-commit-go-commit}|
    $\wedge\ \textit{go\_abort}' = \textit{go\_abort}$|\label{line:go-commit-go-abort}|
$\textit{VoteYes}(n) :=$|\label{line:vote-yes}|
    $\wedge\ \textit{vote\_yes}' = \textit{vote\_yes} \cup \{n\}$|$\label{line:vote-yes-vote-yes}$|
    $\wedge\ \textit{go\_commit}' = \textit{go\_commit}$
    $\wedge\ \textit{go\_abort}' = \textit{go\_abort}$
\end{lstlisting}
\end{minipage}
\hfill
\begin{minipage}[b]{0.48\textwidth}
\begin{lstlisting}[mathescape,numbers=left,xleftmargin=20pt,escapechar=|,firstnumber=\thelstnumber]
$\Init :=$|\label{line:init}|
    $\wedge\ \textit{vote\_yes} = \emptyset$
    $\wedge\ \textit{go\_commit} = \emptyset$ 
    $\wedge\ \textit{go\_abort} = \emptyset$
$\Next :=$|\label{line:next}|
    $\vee\ \textit{GoCommit}$|\label{line:next-go1}|
    $\vee\ \exists n \in \textit{Node} : \textit{VoteYes}(n)$|\label{line:next-vote-yes}|
\end{lstlisting}
\end{minipage}
\vspace{-1em}
\caption{An example of a \TLA protocol (excerpt).}
\vspace{-1.5em}
\label{fig:protocol_example}
\end{figure}

\subsubsection{Modeling Conventions}

We adopt  standard conventions on the syntax used to represent protocols,
particularly on how \Next is written. Specifically, we decompose \Next into a
disjunction of {\em actions} (e.g. \Fig~\ref{fig:protocol_example}
lines~\ref{line:next}-\ref{line:next-vote-yes}). An action is a predicate
mapping a pair of states to true or false; e.g., action \textit{GoCommit} of
\Fig~\ref{fig:protocol_example}. We  decompose an action into the conjunction
of a {\em pre-condition} and a {\em post-condition}. A pre-condition is a
predicate mapping a state to true or false; if the pre-condition of an action
is satisfied by a state, then we say the action is {\em enabled} at that state.
For instance, \Fig~\ref{fig:protocol_example}
line~\ref{line:go-commit-vote-yes} says that action \textit{GoCommit} is
enabled only when all nodes have voted yes.

We decompose a post-condition into a conjunction of {\em post-clauses}, one for
each state variable. A post-clause determines how its associated state variable
changes when the action is taken. For instance, \Fig~\ref{fig:protocol_example}
line~\ref{line:go-commit-go-commit} shows a post-clause for the state variable
\textit{go\_commit}, denoted by priming the variable name:
$\textit{go\_commit}'$.

If $s \rightarrow t$ is a transition and $(s,t)$ satisfies an action $A$, we
can say that {\em $A$ is taken} and write $s \xrightarrow{A} t$. Note that
$(s,t)$ may satisfy multiple actions and we may annotate the transition with
any of them. In this way, runs of a protocol may be outfitted with a sequence
of actions. Annotating runs of a protocol with actions is critical for our
synthesis algorithm, since annotations allow us to ``blame'' particular actions
for causing a counterexample run. 


\subsubsection{Protocol Synthesis}


A tuple $\angles{\Vars,\Holes,\Init,\Next_0}$ is a {\em protocol sketch}, where
\Vars and \Init are as in a \TLA protocol and $\Next_0$ is a transition
relation predicate containing the hole names found in \Holes. \Holes is a
finite (possibly empty) set of tuples, each containing a hole name $h$, a list
of argument symbols $\vec v_h$, and a grammar $G_h$. A hole represents an
uninterpreted function over the arguments $\vec v_h$. Each hole is associated
with exactly one action $A_h$ and it appears exactly once in that action. The
grammar of a hole defines the set of candidate expressions that can fill the
hole. Because $\Holes$ is finite, we will assume without loss of generality that
all holes are assigned an index to order them: $h_1,...,h_n$.

For example, a sketch can be derived from \Fig~\ref{fig:protocol_example} by
replacing the update of line~\ref{line:vote-yes-vote-yes} with
$\textit{vote\_yes}' = h(\textit{vote\_yes}, n)$, where $h$ is the hole name,
the hole has arguments \textit{vote\_yes} and $n$, 
and the action of the hole is $\textit{VoteYes}(n)$. One possible
grammar for this hole is (in Backus Normal Form): 
$$E ::= \emptyset\ |\ \{n\}\ |\ \textit{vote\_yes}\ |\ (E\cup E)\ |\ (E\cap E)\ |\ (E\setminus E) $$
which generates all standard set expressions over the empty set, the singleton
set $\{n\}$, and the set \textit{vote\_yes}. We note that, in general, each hole
of a sketch may have its own distinct grammar. 

A hole is either a {\em pre-hole} or a {\em post-hole}. If the hole is a
pre-hole, it is a placeholder for a pre-condition of the action. If the hole is
a post-hole, it is a placeholder for the right-hand side of a post-clause of
the action, e.g., as in $\textit{vote\_yes}' = h(\textit{vote\_yes}, n)$, where
$h$ is a post-hole. 
We write $h.\textit{var}$ to denote the variable associated with a post-hole
$h$. For instance, $h.\textit{var} = \textit{vote\_yes}$ for the post-hole $h$
shown prior.
We do not consider synthesis of the initial state predicate
and therefore no holes appear in \Init. The arguments of a hole $h$ may include
any of the the state variables in \Vars. If $h$ is a pre-hole, then it returns
a boolean. If the hole is a post-hole, its type is the same as its associated
variable, e.g., hole $h$ above has the same type as \textit{vote\_yes}.

\label{sec:problem-statements}

A {\em completion} of a sketch is a protocol derived from the sketch by
replacing each hole with an expression from its grammar. If the sketch is clear
from context or irrelevant, we may write the completion as a tuple
$(e_1,\ldots,e_n)$, where $e_i$ is the expression filling the $i$th hole.
Informally, the synthesis task is to find a completion of the protocol that
satisfies a given property. 

\begin{problem}
\label{problem:instance}
Let $\angles{\Vars,\Holes,\Init,\Next_0}$ be a sketch and \prop a
property. If one exists, find a completion of the sketch that satisfies \prop.
Otherwise, recognize that all completions violate \prop.
\end{problem}

We briefly note that Problem~\ref{problem:instance} asks to synthesize
protocols that are only guaranteed to be correct for a particular choice of
parameters (e.g. the number of nodes in the protocol). \cite{egolf2024-arxiv}
takes a similar approach: (1) solve Problem~\ref{problem:instance} and then (2)
use the \TLA Proof System (TLAPS) to show that the synthesized protocol is
correct for all parameter choices. Our work in this paper improves step (1) of
this approach, and we refer the reader to~\cite{egolf2024-arxiv} for details on
step (2). So, when we say ``completion satisfies \prop,'' we only guarantee
that the completion satisfies \prop for the parameter choice used in the
synthesis process.

\section{Our Approach}

The high-level architecture of our approach is shown in
Fig.~\ref{fig:arch-top2}. Our method is a counterexample-guided inductive
synthesis (CEGIS) algorithm, so it coordinates between a learner and verifier.
We use the TLC model checker as our verifier~\cite{tlcmodelchecker}. Our
learner works by coordinating between (1) an expression enumerator, (2) a
counterexample generalizer, and (3) a search space reduction maintainer. The
expression enumerator essentially generates candidate protocols by exploring
the non-terminals of the protocol sketch's grammars in a breadth-first manner.
The search space reduction is used by the expression enumerator to explore the
grammars in a way that avoids generating redundant expressions. The
counterexample generalizer converts counterexamples to logical constraints on
expressions. The expression enumerator uses these {\em pruning constraints}
along with a {\em constraint checker} to eliminate candidate protocols that
exhibit previously encountered counterexamples, before they make it to the
verifier.

\begin{figure}
    \centering
    \includegraphics[width=0.9\textwidth]{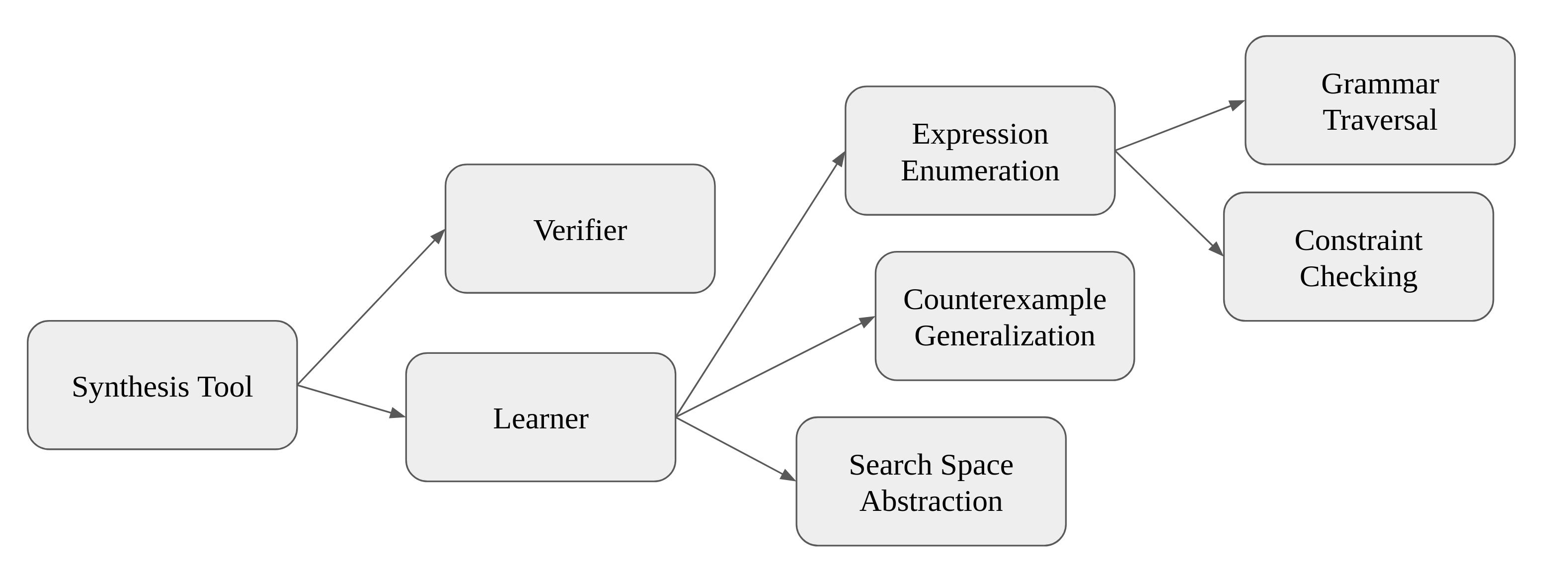}
    \caption{The architecture of our method.}
    \label{fig:arch-top2}
\end{figure}

\begin{figure}
\centering
\begin{minipage}[t]{0.48\textwidth}
\begin{lstlisting}[mathescape,numbers=left,xleftmargin=15pt,escapechar=|]
$\textsc{Synth}(S, \Phi) :=$
  $U := \textit{init\_search\_space}(S)$|\label{line:init-space}|
  $\text{while }\True$:
    $e := \textsc{Pick}(U)$|\label{linepicke}|
    $\text{if } e = \bot \text{ then return } \bot$ |\label{linenosolution}|
\end{lstlisting}
\end{minipage}
\hfill
\begin{minipage}[t]{0.48\textwidth}
\begin{lstlisting}[mathescape,numbers=left,xleftmargin=25pt,escapechar=|,firstnumber=\thelstnumber]
    $r, \textit{correct} := \textsc{Check}(e, \Phi)$|\label{lineverifier}|
    $\text{if } \textit{correct} \text{ then return } e$ |\label{linesolution}|
    $\pi := \textsc{Generalize}(r,e,S)$
    $U := \textsc{Prune}(U,\pi)$|\label{lineprune}|
    $U := \textsc{Abstract}(U, \pi)$|\label{line:abstract}|
\end{lstlisting}
\end{minipage}
\vspace{-1em}
\caption{Our algorithm.}
\vspace{-1.5em}
\label{fig:our-algorithm}
\end{figure}


 Fig.~\ref{fig:our-algorithm} provides procedural details of our approach. The
top-level procedure $\textsc{Synth}$ takes as input a sketch $S$ and a property
$\Phi$. The goal is to find a completion of $S$ that satisfies $\Phi$. In a
nutshell, $\textsc{Synth}$ works by generating candidate expressions $e$ from
the hole grammars (line~\ref{linepicke}) until either the search space $U$
becomes empty, which means there exists no correct completion
(line~\ref{linenosolution}), or until a correct completion is found
(line~\ref{linesolution}). If the current $e$ is incorrect, the counterexample
$r$ returned by the verifier (line~\ref{lineverifier}) is {\em generalized}
into a {\em pruning constraint} $\pi$ which helps prune as well as reduce the
search space (lines~\ref{lineprune}-\ref{line:abstract}).

The \textsc{Pick} procedure generates expressions that (1) do not exhibit
previously encountered counterexamples and (2) are distinct up to
interpretation equivalence. Objective (1) is achieved by maintaining a set of
{\em pruning constraints}, while Objective (2) is achieved by maintaining an
{\em equivalence reduction}. 
The \textsc{Pick} procedure is able to check if a candidate expression
satisfies the accumulated pruning constraints, but the \textsc{Generalize}
procedure is responsible for computing the pruning constraints and the 
\textsc{Prune} procedure for applying them to the search space.
Likewise, \textsc{Pick} is able to explore the equivalence classes of
expressions induced by the interpretation reduction, but the \textsc{Abstract}
procedure is responsible for actually maintaining the reduction.

In Section~\ref{sec:pruning} we describe pruning constraints and the
\textsc{Generalize} routine in more detail. Then in
Section~\ref{sec:abstraction} we describe the representation of the
interpretation reduction, how we explore the reduced search space, and how the
algorithm maintains the reduction as it encounters new counterexamples. Finally,
we discuss the soundness, completeness, and termination of our algorithm in
Section~\ref{sec:correctness}.

\section{Pruning Constraints}
\label{sec:pruning}

\subsubsection{Counterexamples}

A counterexample is a finite run $r$ annotated with actions: $s_0
\xrightarrow{A_1} s_1 \ldots \xrightarrow{A_k} s_k$. We support four types of
counterexamples: safety, deadlock, liveness, and stuttering. $r$ is a {\em
safety violation} of protocol $e$ if $r$ is a run of $e$ and $r$ violates a
safety property. $r$ is a {\em deadlock violation} of protocol $e$ if $r$ is a
run of $e$ and no action is enabled in $s_k$. $r$ is a {\em liveness violation}
of protocol $e$ if (1) $r$ is a run of $e$, (2) a {\em loop} of $r$ is an {\em
unfair cycle}, and (3) the infinite run induced by the cycle violates a
liveness property. A loop of $r$ is a suffix of $r$ that starts and ends at the
same state. 
A loop is a {\em strongly fair cycle} if there exists a strongly fair action
$A$ that is not taken in the loop, but is enabled at some state in the loop.
A loop is a {\em weakly fair cycle} if there exists a weakly fair action $A$
that is not taken in the loop, but is enabled at all states in the loop.
A loop is an unfair cycle if it neither a strongly, nor weakly fair cycle. A
run $r$ is a {\em stuttering violation} of protocol $e$ if (1) $r$ is a run of
$e$, (2) all actions that are fair and enabled in $s_k$ transition back to
$s_k$ (they are self-loops), (3) the infinite run induced by a self-loop at
$s_k$ violates a liveness property, and (4) $s_k$ is not a deadlock state.

\subsubsection{Pruning Constraints}
The syntax of pruning constraints is as follows, where $h$ is any hole name,
$\alpha$ is an assignment of values to the arguments of $h$, and $c$ is a
constant value:
%
$
    \pi ::= d\ |\ \pi\lor\pi\ |\ \pi\land\pi
    \text{ and }
    d ::= h(\alpha) \neq c.
$
%
A completion $(e_1,...,e_n)$ satisfies a pruning constraint of the form $\pi_1
\lor \pi_2$ if and only if it satisfies $\pi_1$ or $\pi_2$. Likewise for
$\pi_1\land\pi_2$, but satisfying both $\pi_1$ and $\pi_2$. The completion
satisfies a pruning constraint of the form $h_i(\alpha) \neq c$ if and only if
$e_i(\alpha) \neq c$, where $e(\alpha)$ denotes the result of substituting the
values of $\alpha$ into $e$.

$\textsc{Generalize}(r,e,S)$ constructs a pruning constraint as a function of
the counterexample $r$, the completion $e$, and the sketch $S$. We assume that
the expression in $e$ for hole $h$ is $e_h$. We also assume that the pre- and
post-holes of $S$ are $\textit{pre}(S)$ and $\textit{post}(S)$ respectively.
Finally, we assume access to the expressions for the fixed, non-hole pre- and
post-conditions of the actions in $S$. 

At a high level, our pruning constraints encode (1) which pre-conditions should
be disabled, (2) which pre-conditions should be enabled, and (3) which
post-conditions should evaluate to a different value in future completions. Our
pruning constraints are each a disjunction of clauses that specify these
conditions and hence we use the following three parametric clauses as {\em gadgets}:
\begin{align*}
    \chi_\textit{move}(\Acts,\Xi,e) &:=
        \bigvee_{A\in\Acts} \bigvee_{s\in\Xi} \bigvee_{h\in\textit{post}(A)} h(s)\neq e_h(s)\\
    \chi_\textit{disable}(\Acts,\Xi) &:=
    \bigvee_{A\in\Acts} \bigvee_{s\in\Xi} \bigvee_{\ h\in\textit{pre}(A)} h(s)\neq\True\\
    \chi_\textit{enable}(\Acts,\Xi) &:=
        \bigvee_{A\in\Acts} \bigwedge_{s\in\Xi} \bigwedge_{\ h\in\textit{pre}(A)} h(s)\neq\False
\end{align*}
Intuitively, the gadget $\chi_\textit{move}$ says that some action $A$ directs
a predecessor state $s$ to some successor state $s'$ that is different from the
one specified by the completion $e$. Formally, $\chi_\textit{move}(\Acts, \Xi,
e)$ says there must exist an action $A$ in $\Acts$, a state $s$ in $\Xi$, and
a post-hole $h$ in $A$ such that $h(s)\neq e_h(s)$. Note: a state $s$ is an
assignment of state variables to values, $e_h$ is the expression chosen for
hole $h$, and $e_h(s)$ is the result of substituting $s$ into $e_h$.

The gadget $\chi_\textit{disable}(\Acts, \Xi)$ says there must exist an action
$A$ in $\Acts$, a state $s$ in $\Xi$, and a pre-hole $h$ in $A$ such that $h$
does not evaluate to $\True$ at state $s$---i.e. $A$ is disabled at $s$ by
virtue of $h$. Finally, $\chi_\textit{enable}(\Acts, \Xi)$ says that there
must exist an action $A$ in $\Acts$ such that for all states $s$ in $\Xi$ and
all pre-holes $h$ in $A$, $h$ does not evaluate to $\False$ at $s$---i.e. $A$
is enabled at all states in $\Xi$ if non-hole pre-conditions allow.

\subsubsection{Generalization of Safety Violations}
Then for a safety counterexample $r = s_0 \xrightarrow{A_1} s_1 \ldots
\xrightarrow{A_k} s_k$ of completion $e$ of sketch $S$, we construct the pruning
constraint $\pi_\textit{safe}(r,e,S)$ as

$$\pi_\textit{safe} := \bigvee_{i=0}^{k-1} \chi_\textit{move}(\{A_{i+1}\}, \{s_i\},
e) \vee \chi_\textit{disable}(\{A_{i+1}\}, \{s_i\})$$ 
which says that we need to move a transition to a different state or disable a
transition in order to avoid the safety violation $r$.

\subsubsection{Generalization of Deadlock Violations}
For a deadlock counterexample, we construct the pruning constraint
$\pi_\textit{dead}(r,e,S)$ as
$$\pi_\textit{dead} :=
\pi_\textit{safe}\vee\chi_\textit{enable}(\Acts_{\textit{dead}}, \{s_k\})$$
where $\Acts_{\textit{dead}}$ is the set of {\em can-enable} actions in sketch
$S$ at state $s_k$. An action $A$ is can-enable at $s_k$ if all (non-hole)
pre-conditions of $A$ evaluate to \True at $s_k$. This constraint says that we
either need to disable the path to the deadlocked state or we need to enable an
action that undeadlocks $s_k$.

\subsubsection{Generalization of Liveness Violations}
For a liveness counterexample, we construct the pruning constraint
$\pi_\textit{live}(r,e,S)$ as
$$\pi_\textit{live} := 
\pi_\textit{safe} 
\vee \chi_\textit{enable}(\Acts_{\textit{weak}}, \textit{cycle}(r))
\vee \bigvee_{s\in\textit{cycle}(r)}\chi_\textit{enable}(\Acts_{\textit{strong}}^{(s)}, \{s\})
$$
where $\textit{cycle}(r)$ is the set of states in the loop of $r$, and
$\Acts_{\textit{strong}}^{(s)}$ is the set of strongly fair, disabled,
can-enable actions in $S$ at state $s$. This constraint says that we need to
either (1) alter the lasso $r$, (2) enable a strongly fair action in some state
of the loop, or (3) enable a weakly fair action in all states of the loop.

\subsubsection{Generalization of Stuttering Violations}

Finally, for a stuttering counterexample, we construct the pruning constraint
$\pi_\textit{stut}(r,e,S)$ as
$$\pi_\textit{stut} :=
\pi_\textit{safe}\vee\chi_\textit{move}(\Acts_\textit{enabled},
\{s_k\}, e)\vee\chi_\textit{enable}(\Acts_\textit{disabled}, \{s_k\})$$
where $\Acts_\textit{enabled}$ is the set of enabled fair actions in state
$s_k$ and $\Acts_\textit{disabled}$ is the set of (strongly or weakly) fair,
disabled, can-enable actions in state $s_k$. Because $s_k$ is a stuttering
state, we know that all enabled fair actions are stuttering on $s_k$ and that
we can avoid the stuttering violation by having one of these actions move to a
state other than $s_k$. We can also avoid the stuttering violation by enabling
a fair action that is disabled at $s_k$. 

Finally, we define $\textsc{Generalize}(r,e,S)$ to be
$\pi_\textit{vtype}(r,e,S)$ where
$\textit{vtype}\in\{\textit{safe,dead,live,stut}\}$ is the violation type of
$r$. 

\begin{lemma}
    \label{lem:pruning-safe}
    Let $r$ be a safety counterexample of completion $e_1$ of sketch $S$. Let
    $\pi = \pi_\textit{safe}(r,e_1,S)$. Then for all completions $e_2$ of $S$,
    $e_2$ satisfies $\pi$ if and only if $r$ is not a safety counterexample of
    $e_2$.
\end{lemma}
\ifdefined\extendedversion
\begin{proof}

Suppose $e_2$ satisfies $\pi$. Then for some $i\in\{0,...,k-1\}$, $e_2$
satisfies at least one of $\chi_\textit{move}(\{A_{i+1}\}, \{s_i\}, e_1)$ or
$\chi_\textit{disable}(\{A_{i+1}\}, \{s_i\})$. In the first case, $e_2$ cannot
transition from $s_i$ to $s_{i+1}$ using action $A_{i+1}$ because the
post-condition does not allow $s_{i+1}$ as a successor state. In the second
case, $A_{i+1}$ is disabled at $s_i$ in $e_2$. Therefore, in either case, the
transition $s_i\xrightarrow{A_{i+1}}s_{i+1}$ does not exist in the protocol
$e_2$, and so $r$ is not a counterexample of $e_2$.

Now if $e_2$ does not satisfy $\pi$, for all $i\in\{0,...,k-1\}$, $e_2$
violates both $\chi_\textit{move}(\{A_{i+1}\}, \{s_i\}, e_1)$ and
$\chi_\textit{disable}(\{A_{i+1}\}, \{s_i\})$. In other words, for all $i$,
$A_{i+1}$ transitions from $s_i$ to $s_{i+1}$ in $e_2$ and is enabled at $s_i$.
Therefore, every transition in $r$ exists in $e_2$, and so $r$ is a safety
counterexample of $e_2$. \qed

\end{proof}

\else
\fi

\begin{lemma}
    \label{lem:pruning-dead}
    Let $r$ be a deadlock counterexample of completion $e_1$ of sketch $S$. Let
    $\pi = \pi_\textit{dead}(r,e_1,S)$. Then for all completions $e_2$ of $S$,
    $e_2$ satisfies $\pi$ if and only if $r$ is not a deadlock counterexample
    of $e_2$.
\end{lemma}
\ifdefined\extendedversion
\begin{proof}

Suppose $e_2$ satisfies $\pi$. Then $e_2$ satisfies one of
\begin{align*}
&\pi_\textit{safe}(r,e_1,S)&
\\
&\chi_\textit{enable}(\Acts_{\textit{dead}},\{s_k\})&
\end{align*}
In the first case, $r$ is not a counterexample of $e_2$ because the
path to $s_k$ is disabled, as per the safety argument above. In the second
case, $e_2$ satisfies all pre-conditions of some can-enable action at $s_k$,
and so $s_k$ is not a deadlock state in $e_2$. Therefore, $r$ is not a deadlock
counterexample of $e_2$.

Suppose $e_2$ does not satisfy $\pi$. Then $e_2$ violates both
\begin{align*}
&\pi_\textit{safe}(r,e_1,S)&
\\
&\chi_\textit{enable}(\Acts_{\textit{dead}}, \{s_k\})&
\end{align*}
Violating the first clause ensures that all transitions in $r$ exist
in $e_2$, and violating the second clause ensures that no can-enable actions
are enabled at $s_k$. Therefore, $r$ is a deadlock counterexample of $e_2$.
\qed

\end{proof}

\else
\fi

\begin{lemma}
    \label{lem:pruning-live}
    Let $r$ be a liveness counterexample of completion $e_1$ of sketch $S$. Let
    $\pi = \pi_\textit{live}(r,e_1,S)$. Then for all completions $e_2$ of $S$,
    $e_2$ satisfies $\pi$ if and only if $r$ is not a liveness counterexample
    of $e_2$.
\end{lemma}
\ifdefined\extendedversion
\begin{proof}
Because $r$ is a liveness violation, let $\textit{cycle}(r) =
s_n\xrightarrow{A_{n+1}}...\xrightarrow{A_k}s_k$ be the loop of the
counterexample, with $s_n = s_k$.

Suppose $e_2$ satisfies $\pi$. Then $e_2$ satisfies at least one of
\begin{align*}
&\pi_\textit{safe}(r,e_1,S)
\\
&\chi_\textit{enable}(\Acts_{\textit{weak}}, \textit{cycle}(r))
\\
&\chi_\textit{enable}(\Acts_{\textit{strong}}^{(s)}, \{s\}) \text{ for some
$s\in\textit{cycle}(r)$}
\end{align*}
In the first case, $r$ is not a
counterexample of $e_2$ because either the path to the loop is disabled or the
loop itself is disabled. In the second case, $e_2$ enables a previously
disabled weakly fair action in all states of the loop, so the loop is unfair.
In the third case, $e_2$ enables a previously disabled strongly fair action in
some state of the loop, so the loop is unfair. Therefore, $r$ is not a liveness
counterexample of $e_2$.

Suppose $e_2$ does not satisfy $\pi$. Then $e_2$ violates all of
\begin{align*}
&\pi_\textit{safe}(r,e_1,S)
\\
&\chi_\textit{enable}(\Acts_{\textit{weak}}, \textit{cycle}(r))
\\
&\chi_\textit{enable}(\Acts_{\textit{strong}}^{(s)}, \{s\}) \text{ for all
$s\in\textit{cycle}(r)$}
\end{align*}
In other words, all transitions in $r$
exist in $e_2$, each previously disabled weakly fair action is still disabled
in at least one state of the loop, and all previously disabled strongly fair
actions are still disabled in all states of the loop. Therefore, $r$ is a
liveness counterexample of $e_2$. Hence, $e_2$ does not satisfy $\Phi$.
\qed

\end{proof}

\else
\fi

\begin{lemma}
    \label{lem:pruning-stut}
    Let $\Phi$ be a property. Let $r$ be a stuttering counterexample of
    completion $e_1$ of sketch $S$. Let $\pi = \pi_\textit{stut}(r,e_1,S)$.
    Then $e_1$ does not satisfy $\pi$. Furthermore, for all completions $e_2$
    of $S$, if $e_2$ does not satisfy $\pi$, then $e_2$ does not satisfy
    $\Phi$.
\end{lemma}
\ifdefined\extendedversion
\begin{proof}

By construction, $e_1$ does not satisfy $\pi$. Specifically, it violates all of
\begin{align*}
&\pi_\textit{safe}(r,e_1,S)&
\\
&\chi_\textit{move}(A_{\textit{enabled}}, \{s_k\}, e_1)&
\\
&\chi_\textit{enable}(A_{\textit{disabled}}, \{s_k\}) &
\end{align*}
It violates $\pi_\textit{safe}$ because $r$ is a run of $e_1$. It violates
$\chi_\textit{move}(A_{\textit{enabled}}, \{s_k\}, e_1)$ because $r$ is a
counterexample of $e_1$ and therefore all enabled fair actions at $s_k$
self-loop on $s_k$. Likewise, it violates $\chi_\textit{enable}
(A_{\textit{disabled}}, \{s_k\})$ because all actions in
$A_{\textit{disabled}}$ are disabled at $s_k$.

Now suppose $e_2$ does not satisfy $\pi$. Then $e_2$ violates all of
\begin{align*}
&\pi_\textit{safe}(r,e_1,S)&
\\
&\chi_\textit{move}(A_{\textit{enabled}}, \{s_k\}, e_1)&
\\
&\chi_\textit{enable}(A_{\textit{disabled}}, \{s_k\})&
\end{align*}
In other words,
$e_2$ all transitions in $r$ exist in $e_2$, all previously enabled fair
actions at $s_k$ self-loop on $s_k$, and all previously disabled actions at
$s_k$ are still disabled. Therefore, if $s_k$ is not a deadlock state in $e_2$,
then $r$ is a stuttering counterexample of $e_2$. Otherwise, $s_k$ is a
deadlock state in $e_2$. In either case, $e_2$ does not satisfy $\Phi$.
\qed

\end{proof}

\else
\fi

\begin{theorem}
    \label{thm:pruning-forward}
    Let $r$ be a counterexample of completion $e_1$ of sketch $S$. Let $\pi =
    \textsc{Generalize}(r,e_1,S)$. Then $e_1$ does not satisfy $\pi$. 
    Furthermore,
    if $r$ is not a stuttering counterexample, then for all completions $e_2$ of
    $S$, if $e_2$ satisfies $\pi$, then $r$ is not a counterexample of $e_2$.
\end{theorem}
\begin{proof}
    This follows from Lemmas~\ref{lem:pruning-safe}-\ref{lem:pruning-stut}.
    \ifdefined\extendedversion
    $r$ is a safety, deadlock, liveness, or stuttering counterexample, so we
procede by case analysis.

\textbf{Safety.} If $r$ is a safety counterexample, then $\pi =
\pi_\textit{safe}$. Therefore, by Lemma~\ref{lem:pruning-safe}, if some
completion $e_2$ satifies $\pi$, then $r$ is not a counterexample of $e_2$.
Since $r$ is a counterexample of $e_1$, we can conclude that $e_1$ does not
satisfy $\pi$.

\textbf{Deadlock.} If $r$ is a deadlock counterexample, then $\pi =
\pi_\textit{dead}$. Therefore, by Lemma~\ref{lem:pruning-dead}, if some
completion $e_2$ satisfies $\pi$, then $r$ is not a counterexample of $e_2$.
Since $r$ is a counterexample of $e_1$, we can also conclude that $e_1$ does
not satisfy $\pi$.

\textbf{Liveness.} If $r$ is a liveness counterexample, then $\pi =
\pi\textit{live}$. Therefore, by Lemma~\ref{lem:pruning-live}, if some
completion $e_2$ satisfies $\pi$, then $r$ is not a counterexample of $e_2$.
Since $r$ is a counterexample of $e_1$, we can also conclude that $e_1$ does
not satisfy $\pi$.

\textbf{Stuttering.} For this case, we only need to show $e_1$ does not satisfy
$\pi_\textit{stut}$. This follows from a direct application of
Lemma~\ref{lem:pruning-stut}.

    \else
    See~\cite{full-tacas} for proofs.
    \fi
\qed
\end{proof}

\begin{theorem}
    \label{thm:pruning-backward}
    Let $\Phi$ be a property. Let $r$ be a counterexample to $\Phi$ of
    completion $e_1$ of sketch $S$. Let $\pi = \textsc{Generalize}(r,e_1,S)$.
    Then for all completions $e_2$ of $S$, if $e_2$ does not satisfy $\pi$,
    then $e_2$ does not satisfy $\Phi$.
    Futhermore, if $r$ is not a stuttering
    counterexample and $e_2$ does not satisfy $\pi$, then $r$ is a
    counterexample of $e_2$.
\end{theorem}
\begin{proof}
    This follows from Lemmas~\ref{lem:pruning-safe}-\ref{lem:pruning-stut}.
    \ifdefined\extendedversion
    $r$ is a safety, deadlock, liveness, or stuttering counterexample, so we
procede by case analysis.

\textbf{Safety.} If $r$ is a safety counterexample, then $\pi =
\pi_\textit{safe}$. Therefore, by Lemma~\ref{lem:pruning-safe}, if some
completion $e_2$ does not satisfy $\pi$, then $r$ is a counterexample of $e_2$.
It follows immediately then that $e_2$ does not satisfy $\Phi$.

\textbf{Deadlock.} If $r$ is a deadlock counterexample, then $\pi =
\pi_\textit{dead}$. Therefore, by Lemma~\ref{lem:pruning-dead}, if some
completion $e_2$ does not satisfy $\pi$, then $r$ is a counterexample of $e_2$.
It follows immediately then that $e_2$ does not satisfy $\Phi$.

\textbf{Liveness.} If $r$ is a livness counterexample, then $\pi =
\pi_\textit{live}$. Therefore, by Lemma~\ref{lem:pruning-live}, if some
completion $e_2$ does not satisfy $\pi$, then $r$ is a counterexample of $e_2$.
It follows immediately then that $e_2$ does not satisfy $\Phi$.

\textbf{Stuttering.} For this case, we only need to show that for all
completions $e_2$ of sketch $S$, if $e_2$ does not satisfy $\pi$, then $e_2$
does not satisfy $\Phi$. This follows from a direct application of
Lemma~\ref{lem:pruning-stut}.

    \else
    See~\cite{full-tacas} for proofs.
    \fi
\qed
\end{proof}


\subsubsection{Exactness of Pruning Constraints}

A pruning constraint $\pi$ is {\em under-pruning} with respect to a sketch $S$
and a counterexample $r$ if there exists a completion $e$ of $S$ such that $e$
satisfies $\pi$ but $r$ is a counterexample of $e$. A pruning constraint $\pi$
is {\em over-pruning} with respect to a sketch $S$ and a counterexample $r$ if
there exists a completion $e$ of $S$ such that $e$ does not satisfy $\pi$ but
$r$ is not a counterexample of $e$. A pruning constraint $\pi$ is {\em exact}
(called ``optimal'' in~\cite{egolf2024-arxiv}) if it is neither under-pruning
nor over-pruning. The safety, deadlock, and liveness pruning constraints
$\pi_\textit{safe}$, $\pi_\textit{dead}$, and $\pi_\textit{live}$, presented
above, are all exact. In contrast, only the safety pruning constraints in prior
work~\cite{egolf2024-arxiv} are exact. The stuttering constraints
$\pi_\textit{stut}$ presented here are not exact, but they are sufficient for
the correctness of our synthesis algorithm (see Section~\ref{sec:correctness}).
\ifdefined\extendedversion
In Appendix~\ref{sec:exact-stut}, we provide an alternative to
$\pi_\textit{stut}$ that is exact.
\else
In the full version of this paper \cite{full-tacas}, we provide an alternative
to $\pi_\textit{stut}$ that is exact. 
\fi
However, this alternative constraint
introduces additional performance overhead without providing sufficient
benefit. Therefore, we do not consider the alternative pruning constraint
further here.

\section{Interpretation Reduction}
\label{sec:abstraction}

In this section we discuss how we represent, use, and maintain the reduced
search space. We conclude the section by showing that our algorithm does not
exclude any correct completions, which is always a risk of reduction
techniques. In Section~\ref{sec:correctness} we show that our algorithm is
sound, complete, and terminating.

\subsubsection{Representation of the Reduced Search Space}

For simplicity, we will temporarily assume that the sketch $S$ has a single
hole $h$. In general, we maintain a reduced search space for each hole in the
sketch. We represent the (reduced) search space $U$
(Fig.~\ref{fig:our-algorithm}) as a tuple $\angles{G, \mathcal A, V, \Pi}$,
where $G$ is the grammar of the hole, $\mathcal A$ is a list of
interpretations, $V$ is a partial map from {\em annotated non-terminals} of $G$
to expressions, and $\Pi$ is the conjunction of all accumulated pruning
constraints.

We write $e_1\equiv_\alpha e_2$ if $e_1$ and $e_2$ are interpretation
equivalent. If $\mathcal A$ is a set of interpretations, then $e_1$ and $e_2$
are interpretation equivalent with respect to $\mathcal A$ if for all $\alpha
\in \mathcal A$, $e_1\equiv_\alpha e_2$.

Intuitively, annotated non-terminals ``label'' the equivalence classes of
expressions induced by interpretation equivalence under $\mathcal A$. Formally,
an annotated non-terminal is a pair $\angles{q, \vec c}$, where $q$ is a
non-terminal in the grammar $G$, and $\vec c = (c_1, \ldots, c_n)$ is a
(possibly empty) tuple of constants. The number of constants $n$ is equal to
the number of interpretations in $\mathcal A$.
We write $\sem{q, \vec
c}_{G,\mathcal A}$ to denote the equivalence class labeled by the annotated
non-terminal $\angles{q, \vec c}$ and omit $G$ and $\mathcal A$ if they are
clear from context. 

If $\alpha_i$ is the $i$-th interpretation in $\mathcal A$ and $q$ is a
non-terminal in the grammar $G$, we write $e\in\sem{q, \vec c}$ if and only if
(1) for all $i$, $e$ evaluates to $c_i$ under $\alpha_i$ and (2) $e$ is
generated by $q$ in $G$. The partial map $V$ has the following invariant: for
every annotated non-terminal $\angles{q, \vec c}$ in the domain of $V$, the
expression $e = V[\angles{q, \vec c}]$ is such that $e\in\sem{q, \vec c}$.
Intuitively, $V$ keeps track of (1) which equivalence classes have been visited
and (2) which enumerated expression represents that equivalence class.

\subsubsection{Enumerating Expressions}

Recall that the \textsc{Pick} procedure is responsible for generating candidate
completions (Fig.~\ref{fig:our-algorithm}, line~\ref{linepicke}). \textsc{Pick} 
works by treating $V$ as a cache of enumerated expressions
and using the rules of the hole grammar $G$ to enumerate larger expressions.
For instance, if $G$ has a rule $q\to (q_1 + q_2)$, and $V[\angles{q_1, (1,
0)}] = x$ and $V[\angles{q_2, (1, 0)}] = y$, then \textsc{Pick} will build the
expression $x + y$ and set $V[\angles{q, (2, 0)}] = x + y$. A critical detail
is that we do not change $V$ if $\angles{q, (2, 0)}$ is already in the domain
of $V$. If an annotated non-terminal is already in $V$, the algorithm 
has already enumerated an expression that is interpretation equivalent to $x +
y$ under $\mathcal A$ for the non-terminal $q$.

\subsubsection{Ensuring Completeness}

Using interpretation reduction allows for improved performance, but we have to
be careful to ensure that we are not erroneously excluding completions. If
$\mathcal A$ does not contain the ``appropriate'' interpretations, then
\textsc{Pick} may return $\bot$ prematurely. Consider a simple example of what
can go wrong. Suppose $G$ is the grammar $E ::= x\ |\ E + 1$ and that we have
the property $\Phi := e \neq 0$. $\Phi$ is violated by $e := x$, since $e$
evaluates to 0 under the interpretation $[x\mapsto 0]$. Now suppose $[x\mapsto
0]$ is not in $\mathcal A$, particularly suppose $\mathcal A$ is empty. Then
\textsc{Pick} will set $V[\angles{E, ()}] = x$ and then try to use the rule
$E\to E + 1$ to generate $x + 1$. However, $x + 1$ and $x$ are interpretation
equivalent under the empty set of interpretations, so \textsc{Pick} will return
$\bot$ before enumerating $x + 1$. $x + 1$ satisfies $\Phi$, so the reduction
has eliminated a valid completion.

We maintain the invariant that $\mathcal A$ contains all interpretations that
appear in the pruning constraints. Specifically, $\textit{interp}(\pi)$ is
defined inductively on the structure of pruning constraints:
$\textit{interp}(h(\alpha) \neq c) = \{ \alpha \}$, and $\textit{interp}(\pi_1
\vee \pi_2) = \textit{interp}(\pi_1\land\pi_2) = \textit{interp}(\pi_1) \cup
\textit{interp}(\pi_2)$. 
Without loss of generality, assume that $\textit{interp}(\Pi)$ is an ordered
list with no duplicates. The expressions in $\textit{interp}(\Pi)$ are
sufficient to ensure the reduction does not exclude any correct completions
(see reasoning across Lemma~\ref{lem:interp-equiv},
Lemma~\ref{lem:pick-complete}, and Theorem~\ref{thm:completeness}).

\begin{lemma}
    \label{lem:interp-equiv}
    Suppose that $e_1\equiv_\mathcal{A}e_2$ with $\mathcal A
    = \textit{interp}(\Pi)$. Then $e_1\vDash\Pi$ if and only if $e_2\vDash\Pi$.
\end{lemma}

We also maintain the invariant that if $\angles{q, \vec c}$ is in the domain
of $V$, then $\vec c$ has the same length as $\mathcal A$. This invariant is
maintained by extending $\vec c$ in all annotated non-terminals every time
we add a new interpretation to $\mathcal A$. In particular, we extend $\vec c$
for each $\angles{q, \vec c}$ in the domain of $V$ by computing the value
of $V[\angles{q, \vec c}]$ under any new interpretations.

\subsubsection{Expressions Modulo Interpretation Equivalence}

A key contribution of our work is that our algorithm is provably complete and
terminating, even when the input is an unrealizable synthesis problem. To prove
these facts we first introduce some additional formalism.

We first define what it means for one set of expressions to subsume another set
of expressions, modulo interpretation equivalence. Let $\mathcal A$ be a set of
interpretations. Let $E_1$ and $E_2$ be two sets of expressions. We say that
$E_1$ is a subset of $E_2$ modulo interpretation equivalence, written
$E_1\subseteq_\mathcal A E_2$, if for all $e_1\in E_1$ there exists an $e_2\in
E_2$ such that $e_1\equiv_\mathcal A e_2$. 

Suppose that $\textsc{Pick}$ returns $\bot$. We denote by $\Lang V$ the set of
expressions contained in $V$ after termination and by $\Lang G$ the set of
expressions generated by the grammar $G$. To guarantee completeness, we need to
ensure that we enumerate every expression in $\Lang G$ modulo interpretation
equivalence---i.e., $\Lang G \subseteq_\mathcal A \Lang V$. Our algorithm for
populating $V$ starts with the terminals of $G$ and works its way up to larger
expressions, per standard grammar enumeration techniques. Therefore, if we
omit the interpretation reduction, our algorithm will enumerate all
expressions in $G$. The only time we might miss an expression is if we ignore
it because it is interpretation equivalent to an expression we have already
enumerated. Furthermore, if two expressions are interpretation equivalent, they
are equally useful for generating new expressions that are distinct up to
interpretation equivalence. Hence, $\Lang G\subseteq_\mathcal A\Lang V$. The
following corollary follows immediately from this observation.

\begin{corollary}
    \label{cor:mod-interp-equiv}
    If $e_1\in\Lang G\setminus\Lang V$, then there exists an $e_2\in\Lang G\cap\Lang
    V$ such that $e_1$ is interpretation equivalent to $e_2$ under $\mathcal A$.
\end{corollary}

\begin{lemma}
    \label{lem:reduction-complete}
    Suppose that if $e\not\vDash\Pi$, then $e\not\vDash\Phi$. Also suppose
    that $\mathcal A = \textit{interp}(\Pi)$.
    Finally, suppose that for all $e\in\Lang G\cap\Lang V$, $e\not\vDash\Pi$.
    Then for all $e\in\Lang G$, $e\not\vDash\Phi$.
\end{lemma}
\begin{proof}
    Suppose that $e_1\in\Lang G$. Either $e_1\in\Lang V$ or $e_1\not\in\Lang
    V$. If $e_1\in\Lang V$, then $e_1\not\vDash\Pi$ and hence
    $e_1\not\vDash\Phi$. Otherwise, $e_1\notin\Lang V$ and by
    Corollary~\ref{cor:mod-interp-equiv}: $e_1\equiv_\mathcal A e_2$ for some
    $e_2\in\Lang G\cap\Lang V$. By assumption, $e_2\not\vDash\Pi$ and therefore
    by Lemma~\ref{lem:interp-equiv}, $e\not\vDash\Pi$. So $e\not\vDash\Phi$. 
\qed
\end{proof}

\begin{lemma}
    \label{lem:pick-complete}
    If $\textsc{Pick}(U) = \bot$, then there is no completion of the sketch $S$
    that satisfies the property $\Phi$.
\end{lemma}
\begin{proof}
    This lemma follows from Theorem~\ref{thm:pruning-backward} and
    Lemma~\ref{lem:reduction-complete}, which say, respectively, that neither
    the pruning constraints nor the interpretation reduction exclude any
    correct completions.
\qed
\end{proof}

\section{Soundness, Completeness, and Termination}
\label{sec:correctness}

The soundness of our method is relatively straightforward.
\begin{theorem} 
\label{thm:soundness} 
If $\textsc{Synth}(S,\Phi) = e$ and $e\neq\bot$, then $e$ is a completion of the sketch
$S$ and $e$ satisfies the property $\Phi$. 
\end{theorem}
\begin{proof}
This theorem follows from (1) expressions chosen by the \textsc{Pick} subroutine come
from the hole grammars of the sketch and (2) every
completion is model checked against $\Phi$ before being returned.
\qed
\end{proof}

The completeness of our method follows from the completeness of the
\textsc{Pick} subroutine.
\begin{theorem}
\label{thm:completeness}
If $\textsc{Synth}(S,\Phi) = \bot$, then there is no completion of the sketch $S$ that
satisfies the property $\Phi$.
\end{theorem}
\begin{proof}
    \textsc{Synth} returns $\bot$ only if the \textsc{Pick} subroutine returns
    $\bot$. So by Lemma~\ref{lem:pick-complete}, there is no
    completion of the sketch $S$ that satisfies the property $\Phi$.
\qed
\end{proof}

Finally, under the assumptions specified in the following theorem, our
method is guaranteed to terminate.
\begin{theorem}
\label{thm:termination}
Suppose that for each set of interpretations $\mathcal A$, there are only
finitely many distinct expressions in $G$, up to interpretation equivalence.
Then the \textsc{Synth} algorithm terminates for any input sketch $S$ and
property $\Phi$.
\end{theorem}
\begin{proof}
    We synthesize protocol instances with finite domain state variables, so the
    size of $\mathcal A$ is bounded. So eventually the $\mathcal A$ will be
    updated for the last time. There are only finitely many distinct
    expressions in $G$, up to interpretation equivalence under the final
    $\mathcal A$. Therefore, the \textsc{Pick} subroutine will eventually
    return a correct completion or $\bot$.
\qed
\end{proof}

We remark that Theorem~\ref{thm:termination} does not
contradict the undecidability of Problem~\ref{problem:instance}. We guarantee termination only for
sketches with finitely many distinct expressions, up to interpretation
equivalence. All the benchmarks used in our experiments satisfy this condition.

\section{Evaluation}
\label{sec:evaluation}

We implement our method in a tool called \tool. In this section, we compare
\tool to \scythe, a state-of-the-art distributed protocol synthesis
tool~\cite{egolf2024-arxiv}. \scythe 
is publicly available~\cite{scythe-full-results}. We evaluate the performance
of the two tools on a set of benchmarks taken from~\cite{egolf2024-arxiv}. Our
experiments come from seven benchmark protocols: two phase commit (2PC),
consensus, simple decentralized lock (DL), lock server (LS), sharded key value
store (SKV), and two reconfigurable raft protocols (RR and RR-big). For RR and
RR-big we use the same incomplete protocols as in~\cite{egolf2024-arxiv}.
Otherwise, the easiest benchmarks have all pre- or all post-conditions missing
from one action of the protocol. The hardest benchmarks have all pre- and
post-conditions missing from two actions. \tool and \scythe are both written in
Python and all experiments are run on a dedicated 2.4GHz CPU.

We conduct two types of experiments: (realizable) synthesis experiments, and
unrealizability experiments. In synthesis experiments, we compare the execution
time of \tool and \scythe on realizable synthesis problems. \tool
is faster than \scythe in 160 out of 171 realizable synthesis experiments. 
Of the 11 experiments where \scythe is faster, the difference in runtimes is
more than 10 seconds in just 2 cases.

In the unrealizability experiments, we compare the ability of \tool and \scythe
to detect unrealizability. Our tool was able to detect unrealizability in 80
out of 123 instances and timed out after 1 hour in the remaining 43 instances.
\scythe recognized unrealizability in 16 of the 123 instances and timed out
after 1 hour in the remaining 107 instances.

\subsection{Realizable Synthesis Experiments}

We partition our synthesis experiments into three categories: short run
experiments (Fig.~\ref{fig:exec-time-comparison-short}) with a total execution
time of less than two minutes, longer experiments where both tools found a
solution (Fig.~\ref{fig:exec-time-comparison-no-to}), and still heavier
experiments (Fig.~\ref{fig:exec-time-comparison-to}), in which \scythe always
times out. In these figures, number labels on the horizontal axis correspond to
distinct synthesis problems and every other tick is skipped in
Fig.~\ref{fig:exec-time-comparison-short} for readability.

\begin{figure}
    \centering
    \includesvg[width=1.0\textwidth]{figs/fig4}
    \vspace{-2.25em}
    \caption{Execution time comparison for short run experiments (total
    execution time less than 120 seconds). Only even ticks shown on horizontal axis.}
    \vspace{2.25em}
    \label{fig:exec-time-comparison-short}
    \centering
    \includesvg[width=1.0\textwidth]{figs/fig5}
    \vspace{-2.25em}
    \caption{Execution time comparison for longer run experiments where neither tool
    timed out (1 hour) in all five runs.}
    \vspace{2.25em}
    \label{fig:exec-time-comparison-no-to}
    \centering
    \includesvg[width=1.0\textwidth]{figs/fig6}
    \vspace{-2.25em}
    \caption{Execution times for \tool in experiments where all five runs of
    \scythe timed out after 1 hour. \scythe always times out here, so its bars
    are not shown.} 
    \vspace{2.25em}
    \label{fig:exec-time-comparison-to}
\end{figure}

There are 171 pairs of bars shown across the three figures. Each pair
corresponds to an input to the synthesis tools (a distinct sketch-property
pair). The left bar (gray, dotted) in each pair shows the execution time of
\scythe, and the right bar (black, solid) shows the execution time of \tool. We
ran each of the experiments five times and report the minimum time
for both tools, as if we ran each tool five times in parallel and halted after
the first run halted. In all cases where \tool timed out in all five runs,
\scythe also did.

\tool performed better than \scythe in 160 out of 171 experiments. In
Fig.~\ref{fig:exec-time-comparison-no-to}, indices 11 and 23 show the 2 cases
where \scythe is faster than \tool by more than 10 seconds. These two
experiments where \scythe does much better use a variation of the consensus
protocol sketch. 
Even so, there are several experiments where \tool is faster than \scythe on
the consensus protocol by more than an order of magnitude (e.g., in
Fig.~\ref{fig:exec-time-comparison-no-to} at index 21, \tool is 67 times faster
than \scythe).

Across all experiments in Figs.~\ref{fig:exec-time-comparison-short}
and~\ref{fig:exec-time-comparison-no-to} (cases where both tools succeeded),
the total execution time for \scythe and \tool are about 32,000 and 9,000
seconds respectively. \tool is routinely faster by a factor of 100 or more. 

\scythe times out in all five runs of all experiments of
Fig.~\ref{fig:exec-time-comparison-to}, hence
Fig.~\ref{fig:exec-time-comparison-to} only includes bars for \tool. 
In the
first 6 experiments shown in Fig.~\ref{fig:exec-time-comparison-to}, \tool
finished in less than 36 seconds, while \scythe timed out after 1 hour. In
Fig.~\ref{fig:exec-time-comparison-to} at index 1, \tool finished in 18
seconds, while \scythe timed out after 1 hour; a speedup of at least 200 times. 

The experiment at index 36 of Fig.~\ref{fig:exec-time-comparison-to} is to
synthesize the entire distributed lock (DL) protocol {\em from scratch}. DL
contains 2 pre-conditions and 4 post-conditions across 2 actions, so a total of
6 expressions need to be synthesized. \scythe times out after 1 hour, but \tool
synthesizes the entire protocol in 125 seconds ($<3$ minutes). To our
knowledge, ours is the first tool to synthesize an entire \TLA protocol from
scratch.

\subsection{Unrealizability Experiments}


There are several reasons a synthesis problem might be unrealizable and
these reasons can be divided into two broad categories. In the first category,
the chosen sketch grammar is not expressive enough, but there exists a grammar
that would make the synthesis problem realizable. In the second category, the
fixed parts of the sketch surrounding the holes are such that {\em no sketch
grammar} is expressive enough. 

The second category can be further divided: (1) there aren't enough state
variables, (2) there aren't enough actions, (3) the fixed pre- and
post-conditions are simply wrong, e.g., a pre-condition always evaluates to false,
(4) a parameterized action is missing an argument, etc.

We construct unrealizable synthesis problems across these categories by
transforming our realizable synthesis problems. These transformations include:
(1) removing rules from the sketch grammar, (2) removing state variables from
the sketch, (3) removing actions from the sketch, (4) changing the fixed pre-
and post-conditions, and (5) removing parameters from actions. All of these
transformations reflect omissions that a reasonable user might make when
constructing a sketch.

In total, we ran 171 realizable synthesis problems in the previous section. We
adapted these to obtain 123 unrealizable synthesis problems. There isn't
a one-to-one correspondence between realizable and unrealizable synthesis
problems because some transformations remove actions and these actions have
holes.


%
\tool was able to detect unrealizability in 80 
out of 123 instances and timed out after 1 hour in the remaining 43 instances.
The experiments are partitioned by easy, medium, and hard benchmarks, which
were derived from the experiments in
Figs.~\ref{fig:exec-time-comparison-short},~\ref{fig:exec-time-comparison-no-to},
and~\ref{fig:exec-time-comparison-to} respectively. \tool succeeded in 46/50
easy experiments, 21/35 medium experiments, and 13/38 hard experiments. In
contrast, \scythe recognized unrealizability in 0/50, 9/35, and 7/38 easy, medium, and
hard experiments respectively. \scythe timed out after 1 hour in the remaining
107 instances. 

\tool terminated in all cases where \scythe did. In these cases where both terminated,
the execution times of the tools were within 1 second of each other and they both
terminated in less than 10 seconds. I.e., \scythe only terminated in cases where
the unrealizability was particularly easy to detect.




\section{Related Work}

\cite{ScenariosHVC2014,DBLP:journals/sigact/AlurT17,DBLP:conf/atva/EgolfT23,DBLP:journals/sttt/FinkbeinerS13,DBLP:journals/corr/FinkbeinerT15}
study synthesis of explicit-state machines.
TRANSIT~\cite{DBLP:conf/pldi/UdupaRDMMA13} requires a human in the loop to
handle counterexamples. All of
\cite{DBLP:journals/acta/MirzaieFJB20,DBLP:conf/cav/BloemBJ16,DBLP:conf/tacas/JaberWJKS23,DBLP:conf/opodis/Lazic0WB17}
consider special classes of distributed protocols.
\cite{DBLP:conf/cav/AlurRSTU15} uses an ad-hoc specification language and
relies on an external SyGuS solver to translate explicit input-output tables
representing expressions into symbolic expressions. Works like
\cite{KatzPeled2009} and \cite{ml-based-synthesis} use methods that are not
guaranteed to find a solution even if one exists.

Unrealizability results are not typically reported in the distributed protocol
synthesis literature. For instance, all results reported
in~\cite{ScenariosHVC2014,DBLP:journals/sigact/AlurT17,DBLP:conf/pldi/UdupaRDMMA13,DBLP:journals/acta/MirzaieFJB20,DBLP:conf/cav/BloemBJ16,DBLP:conf/tacas/JaberWJKS23,DBLP:conf/cav/AlurRSTU15,KatzPeled2009,ml-based-synthesis}
are for realizable instances. Some related work reports results
for unrealizable instances, but the target systems are either
explicit-state~\cite{DBLP:conf/atva/EgolfT23,DBLP:journals/sttt/FinkbeinerS13,DBLP:journals/corr/FinkbeinerT15}
or from a special class of distributed
protocols~\cite{DBLP:conf/opodis/Lazic0WB17}.

Prior to our work, \cite{egolf2024-arxiv} is the only work that synthesizes
general purpose, symbolic distributed protocols written in \TLA.
Compared to that of~\cite{egolf2024-arxiv}, our approach has several major
differences.
First, our method uses interpretation equivalence, while~\cite{egolf2024-arxiv}
uses universal equivalence.
Second, our counterexample generalization is exact except for stuttering
counterexamples whereas that of~\cite{egolf2024-arxiv} is only exact
for safety counterexamples (c.f.~Theorems~\ref{thm:pruning-forward}
\&~\ref{thm:pruning-backward}, and related discussion).
Third, our procedure is guaranteed to terminate under given conditions
(satisfied in all our benchmarks), whereas under these same conditions, the
method of~\cite{egolf2024-arxiv} may not terminate (and indeed times out in
many experiments).
Finally, our experimental results show that our approach is empirically better
than that of~\cite{egolf2024-arxiv} both in realizable and unrealizable problem
instances.



Existing SyGuS solvers use SMT formulas to express properties, and are
therefore not directly applicable to distributed protocol synthesis which
requires temporal logic properties. But our techniques for generating
expressions and checking them against pruning constraints are generally related
to term enumeration strategies used in
SyGuS~\cite{DBLP:conf/fmcad/AlurBJMRSSSTU13}. Both
EUSolver~\cite{DBLP:conf/tacas/AlurRU17} and cvc4sy~\cite{RBN19} are SyGuS
solvers that generate larger expressions from smaller expressions. EUSolver
uses divide-and-conquer techniques in combination with decision tree learning
and is quite different from our approach. To our knowledge, EUSolver does not
employ equivalence reduction at all. The ``fast term enumeration strategy'' of cvc4sy
is similar to our cache-like treatment of $V$ and also uses an equivalence reduction
technique.
\cite{DBLP:conf/cav/HuBCDR19,DBLP:conf/pldi/HuCDR20,DBLP:journals/pacmpl/KimDR23,DBLP:journals/corr/abs-2401-13244}
can recognize unrealizable SyGuS problems, but do not handle temporal
logics required for distributed protocol synthesis. To our knowledge, none of
these approaches use anything like interpretation equivalence to reduce the
search space. 
Absynthe \cite{DBLP:conf/nfm/FedchinDFMRRRSW23} uses a fixed
abstraction provided by the user to guide synthesis of programs from source
languages with complex semantics (e.g., Python), albeit not for distributed
protocols. Our abstraction (the interpretation reduction) is automatically
generated and always changing based on the accumulated counterexamples.

\ifdefined\extendedversion
\section{Conclusion}
\else
\vspace{1em}
\section{Conclusion}
\vspace{1em}
\fi

We presented a novel CEGIS-based 
synthesis method for distributed protocols. We demonstrated that our method is able to
synthesize protocols faster than the state of the art: in some cases, by several orders
of magnitude. In one case we were able to synthesize an entire \TLA protocol
from scratch in less than 3 minutes where the state of the art timed out after
an hour. 

We also provided conditions, satisfied by our benchmarks, under which our
method is guaranteed to terminate, even in cases where the synthesis problem
has no solution; the state of the art is not guaranteed to terminate under
these conditions and makes no guarantees about termination in general. In
practice,
our method recognizes five times more unrealizable synthesis instances than the
state of the art. 

Our results are enabled first and foremost by a novel search space reduction
technique called interpretation reduction; we proved that this technique does
not compromise the completeness of the synthesis algorithm. Additionally, we
use an advanced method for generalizing counterexamples.

For future work, we plan to investigate more sophisticated techniques for
traversing the reduced search space. We are also investigating how to synthesize
protocols when the actions and state variables are not known in advance.

\begin{credits}
\subsubsection{\ackname}
This material is partly supported by the National Science
Foundation under Graduate Research Fellowship Grant \#1938052,
and Award \#2319500. Any opinion,
findings, and conclusions or recommendations expressed in this material are
those of the authors(s) and do not necessarily reflect the views of the National
Science Foundation.
\end{credits}

%
%
%
\bibliographystyle{splncs04}
\bibliography{bib}

\begin{thebibliography}{10}
\providecommand{\url}[1]{\texttt{#1}}
\providecommand{\urlprefix}{URL }
\providecommand{\doi}[1]{https://doi.org/#1}

\bibitem{DBLP:conf/fmcad/AlurBJMRSSSTU13}
Alur, R., Bod{\'{\i}}k, R., Juniwal, G., Martin, M.M.K., Raghothaman, M.,
  Seshia, S.A., Singh, R., Solar{-}Lezama, A., Torlak, E., Udupa, A.:
  Syntax-guided synthesis. In: Formal Methods in Computer-Aided Design, {FMCAD}
  2013, Portland, OR, USA, October 20-23, 2013. pp.~1--8. {IEEE} (2013),
  \url{https://ieeexplore.ieee.org/document/6679385/}

\bibitem{ScenariosHVC2014}
Alur, R., Martin, M., Raghothaman, M., Stergiou, C., Tripakis, S., Udupa, A.:
  {Synthesizing Finite-state Protocols from Scenarios and Requirements}. In:
  Haifa Verification Conference. LNCS, vol.~8855. Springer (2014)

\bibitem{DBLP:conf/tacas/AlurRU17}
Alur, R., Radhakrishna, A., Udupa, A.: Scaling enumerative program synthesis
  via divide and conquer. In: Tools and Algorithms for the Construction and
  Analysis of Systems - 23rd International Conference, {TACAS} 2017. Lecture
  Notes in Computer Science, vol. 10205, pp. 319--336 (2017).
  \doi{10.1007/978-3-662-54577-5_18},
  \url{https://doi.org/10.1007/978-3-662-54577-5\_18}

\bibitem{DBLP:conf/cav/AlurRSTU15}
Alur, R., Raghothaman, M., Stergiou, C., Tripakis, S., Udupa, A.: Automatic
  completion of distributed protocols with symmetry. In: Kroening, D.,
  Pasareanu, C.S. (eds.) Computer Aided Verification - 27th International
  Conference, {CAV}. Lecture Notes in Computer Science, vol.~9207, pp.
  395--412. Springer (2015). \doi{10.1007/978-3-319-21668-3_23},
  \url{https://doi.org/10.1007/978-3-319-21668-3\_23}

\bibitem{DBLP:journals/sigact/AlurT17}
Alur, R., Tripakis, S.: Automatic synthesis of distributed protocols. {SIGACT}
  News  \textbf{48}(1),  55--90 (2017). \doi{10.1145/3061640.3061652},
  \url{https://doi.org/10.1145/3061640.3061652}

\bibitem{DBLP:conf/cav/BloemBJ16}
Bloem, R., Braud{-}Santoni, N., Jacobs, S.: Synthesis of self-stabilising and
  byzantine-resilient distributed systems. In: Chaudhuri, S., Farzan, A. (eds.)
  Computer Aided Verification - 28th International Conference, {CAV}. Lecture
  Notes in Computer Science, vol.~9779, pp. 157--176. Springer (2016).
  \doi{10.1007/978-3-319-41528-4_9},
  \url{https://doi.org/10.1007/978-3-319-41528-4\_9}

\bibitem{Buchman2016TendermintBF}
Buchman, E.: Tendermint: Byzantine fault tolerance in the age of blockchains
  (2016), \url{https://api.semanticscholar.org/CorpusID:59082906}

\bibitem{Buterin2013}
Buterin, V.: Ethereum white paper: A next generation smart contract \&
  decentralized application platform  (2013),
  \url{https://github.com/ethereum/wiki/wiki/White-Paper}

\bibitem{corbett2013spanner}
Corbett, J.C., Dean, J., Epstein, M., Fikes, A., Frost, C., Furman, J.J.,
  Ghemawat, S., Gubarev, A., Heiser, C., Hochschild, P., et~al.: Spanner:
  Google’s globally distributed database. ACM Transactions on Computer
  Systems (TOCS)  \textbf{31}(3),  1--22 (2013)

\bibitem{decandia2007dynamo}
DeCandia, G., Hastorun, D., Jampani, M., Kakulapati, G., Lakshman, A., Pilchin,
  A., Sivasubramanian, S., Vosshall, P., Vogels, W.: Dynamo: Amazon's highly
  available key-value store. ACM SIGOPS operating systems review
  \textbf{41}(6),  205--220 (2007)

\bibitem{scythe-full-results}
Egolf, D.: scythe-fmcad2024. \url{https://github.com/egolf-cs/scythe-fmcad2024}

\bibitem{egolf2024-arxiv}
Egolf, D., Schultz, W., Tripakis, S.: Efficient synthesis of symbolic
  distributed protocols by sketching. In: Narodytska, N., Rümmer, P. (eds.)
  Proceedings of the 24th Conference on Formal Methods in Computer-Aided Design
  -- FMCAD 2024. pp. 281--291. TU Wien Academic Press (2024).
  \doi{10.34727/2024/isbn.978-3-85448-065-5_34},
  \url{https://doi.org/10.34727/2024/isbn.978-3-85448-065-5\_34}

\bibitem{DBLP:conf/atva/EgolfT23}
Egolf, D., Tripakis, S.: Synthesis of distributed protocols by enumeration
  modulo isomorphisms. In: {ATVA} 2023 - Part {I}. pp. 270--291. Lecture Notes
  in Computer Science, Springer (2023). \doi{10.1007/978-3-031-45329-8_13},
  \url{https://doi.org/10.1007/978-3-031-45329-8\_13}

\bibitem{DBLP:conf/nfm/FedchinDFMRRRSW23}
Fedchin, A., Dean, T., Foster, J.S., Mercer, E., Rakamaric, Z., Reger, G.,
  Rungta, N., Salkeld, R., Wagner, L., Waldrip, C.: A toolkit for automated
  testing of dafny. In: Rozier, K.Y., Chaudhuri, S. (eds.) {NASA} Formal
  Methods - 15th International Symposium, {NFM} 2023, Houston, TX, USA, May
  16-18, 2023, Proceedings. Lecture Notes in Computer Science, vol. 13903, pp.
  397--413. Springer (2023). \doi{10.1007/978-3-031-33170-1\_24},
  \url{https://doi.org/10.1007/978-3-031-33170-1\_24}

\bibitem{DBLP:journals/sttt/FinkbeinerS13}
Finkbeiner, B., Schewe, S.: Bounded synthesis. Int. J. Softw. Tools Technol.
  Transf.  \textbf{15}(5-6),  519--539 (2013). \doi{10.1007/S10009-012-0228-Z},
  \url{https://doi.org/10.1007/s10009-012-0228-z}

\bibitem{DBLP:journals/corr/FinkbeinerT15}
Finkbeiner, B., Tentrup, L.: Detecting unrealizability of distributed
  fault-tolerant systems. Log. Methods Comput. Sci.  \textbf{11}(3) (2015).
  \doi{10.2168/LMCS-11(3:12)2015},
  \url{https://doi.org/10.2168/LMCS-11(3:12)2015}

\bibitem{2021ic3posymmetry}
Goel, A., Sakallah, K.: {On Symmetry and Quantification: A New Approach to
  Verify Distributed Protocols}. In: NASA Formal Methods: 13th International
  Symposium, NFM 2021. p. 131–150 (2021)

\bibitem{GulwaniPolozovSingh2017}
Gulwani, S., Polozov, O., Singh, R.: Program synthesis. Foundations and Trends
  in Programming Languages  \textbf{4}(1-2),  1--119 (2017).
  \doi{10.1561/2500000010}

\bibitem{2021swisshance}
Hance, T., Heule, M., Martins, R., Parno, B.: {Finding Invariants of
  Distributed Systems: It{\textquoteright}s a Small (Enough) World After All}.
  In: 18th USENIX Symposium on Networked Systems Design and Implementation
  (NSDI 21). pp. 115--131. USENIX Association (Apr 2021),
  \url{https://www.usenix.org/conference/nsdi21/presentation/hance}

\bibitem{DBLP:conf/cav/HuBCDR19}
Hu, Q., Breck, J., Cyphert, J., D'Antoni, L., Reps, T.W.: Proving
  unrealizability for syntax-guided synthesis. In: Dillig, I., Tasiran, S.
  (eds.) Computer Aided Verification - 31st International Conference, {CAV}.
  Lecture Notes in Computer Science, vol. 11561, pp. 335--352. Springer (2019).
  \doi{10.1007/978-3-030-25540-4_18},
  \url{https://doi.org/10.1007/978-3-030-25540-4\_18}

\bibitem{DBLP:conf/pldi/HuCDR20}
Hu, Q., Cyphert, J., D'Antoni, L., Reps, T.W.: Exact and approximate methods
  for proving unrealizability of syntax-guided synthesis problems. In:
  Donaldson, A.F., Torlak, E. (eds.) Proceedings of the 41st {ACM} {SIGPLAN}
  International Conference on Programming Language Design and Implementation,
  {PLDI} 2020, London, UK, June 15-20, 2020. pp. 1128--1142. {ACM} (2020).
  \doi{10.1145/3385412.3385979}, \url{https://doi.org/10.1145/3385412.3385979}

\bibitem{ml-based-synthesis}
Hui, Y., Ripberger, D., Lu, X., Wang, Y.: Learning distributed protocols with
  zero knowledge. In: Machine Learning for Systems at NeurIPS 2023 (2023),
  \url{https://openreview.net/forum?id=u0Ncut8ru5}

\bibitem{DBLP:conf/tacas/JaberWJKS23}
Jaber, N., Wagner, C., Jacobs, S., Kulkarni, M., Samanta, R.: Synthesis of
  distributed agreement-based systems with efficiently-decidable verification.
  In: {TACAS} 2023. Lecture Notes in Computer Science, vol. 13994, pp.
  289--308. Springer (2023),
  \url{https://doi.org/10.1007/978-3-031-30820-8\_19}

\bibitem{DBLP:journals/corr/JacobsB14}
Jacobs, S., Bloem, R.: Parameterized synthesis. Log. Methods Comput. Sci.
  \textbf{10}(1) (2014). \doi{10.2168/LMCS-10(1:12)2014},
  \url{https://doi.org/10.2168/LMCS-10(1:12)2014}

\bibitem{KatzPeled2009}
Katz, G., Peled, D.: Synthesizing solutions to the leader election problem
  using model checking and genetic programming. In: Haifa Verification
  Conference. p. 117–132. HVC'09, Springer (2009)

\bibitem{DBLP:journals/pacmpl/KimDR23}
Kim, J., D'Antoni, L., Reps, T.W.: Unrealizability logic. Proc. {ACM} Program.
  Lang.  \textbf{7}({POPL}),  659--688 (2023). \doi{10.1145/3571216},
  \url{https://doi.org/10.1145/3571216}

\bibitem{lamport2002specifying}
Lamport, L.: {Specifying Systems: The TLA+ Language and Tools for Hardware and
  Software Engineers}. Addison-Wesley (Jun 2002)

\bibitem{DBLP:conf/opodis/Lazic0WB17}
Lazic, M., Konnov, I., Widder, J., Bloem, R.: Synthesis of distributed
  algorithms with parameterized threshold guards. In: 21st International
  Conference on Principles of Distributed Systems, {OPODIS}. LIPIcs, vol.~95,
  pp. 32:1--32:20. Schloss Dagstuhl - Leibniz-Zentrum f{\"{u}}r Informatik
  (2017). \doi{10.4230/LIPICS.OPODIS.2017.32},
  \url{https://doi.org/10.4230/LIPIcs.OPODIS.2017.32}

\bibitem{DBLP:journals/acta/MirzaieFJB20}
Mirzaie, N., Faghih, F., Jacobs, S., Bonakdarpour, B.: Parameterized synthesis
  of self-stabilizing protocols in symmetric networks. Acta Informatica
  \textbf{57}(1-2),  271--304 (2020). \doi{10.1007/S00236-019-00361-7},
  \url{https://doi.org/10.1007/s00236-019-00361-7}

\bibitem{DBLP:journals/corr/abs-2401-13244}
Nagy, S., Kim, J., D'Antoni, L., Reps, T.W.: Automating unrealizability logic:
  Hoare-style proof synthesis for infinite sets of programs. CoRR
  \textbf{abs/2401.13244} (2024). \doi{10.48550/ARXIV.2401.13244},
  \url{https://doi.org/10.48550/arXiv.2401.13244}

\bibitem{NewcombeAmazon2015}
Newcombe, C., Rath, T., Zhang, F., Munteanu, B., Brooker, M., Deardeuff, M.:
  {How Amazon Web Services Uses Formal Methods}. Commun. ACM  \textbf{58}(4),
  66--73 (Mar 2015). \doi{10.1145/2699417},
  \url{http://doi.acm.org/10.1145/2699417}

\bibitem{PnueliRosner89}
Pnueli, A., Rosner, R.: On the synthesis of a reactive module. In: Proceedings
  of the 16th ACM SIGPLAN-SIGACT Symposium on Principles of Programming
  Languages. p. 179–190. POPL '89, Association for Computing Machinery, New
  York, NY, USA (1989). \doi{10.1145/75277.75293},
  \url{https://doi.org/10.1145/75277.75293}

\bibitem{PnueliRosner90}
Pnueli, A., Rosner, R.: Distributed reactive systems are hard to synthesize.
  In: Proceedings of the 31th IEEE Symposium on Foundations of Computer
  Science. pp. 746--757 (1990)

\bibitem{RBN19}
Reynolds, A., Barbosa, H., N{\"o}tzli, A., Tinelli, C., Barrett, C.: {CVC4SY}:
  Smart and fast term enumeration for syntax-guided synthesis. In: Dillig, I.,
  Tasiran, S. (eds.) Proceedings of the 31st International Conference on
  Computer Aided Verification (CAV). Lecture Notes in Computer Science, vol.
  11561, pp. 74--83. Springer (Jul 2019). \doi{10.1007/978-3-030-25543-5_5},
  \url{http://theory.stanford.edu/~barrett/pubs/RBN+19.pdf}

\bibitem{schultz2024scalablearxiv}
Schultz, W., Ashton, E., Howard, H., Tripakis, S.: {Scalable, Interpretable
  Distributed Protocol Verification by Inductive Proof Slicing}. arXiv eprint
  2404.18048 (2024)

\bibitem{DBLP:conf/fmcad/SchultzDT22}
Schultz, W., Dardik, I., Tripakis, S.: {Plain and Simple Inductive Invariant
  Inference for Distributed Protocols in {TLA}\({}^{\mbox{+}}\)}. In: 22nd
  Formal Methods in Computer-Aided Design, {FMCAD} 2022. pp. 273--283. {IEEE}
  (2022). \doi{10.34727/2022/ISBN.978-3-85448-053-2_34},
  \url{https://doi.org/10.34727/2022/isbn.978-3-85448-053-2\_34}

\bibitem{LezamaAPLAS2009}
Solar-Lezama, A.: The sketching approach to program synthesis. In: Proceedings
  of the 7th Asian Symposium on Programming Languages and Systems. pp. 4--13.
  APLAS '09, Springer (2009)

\bibitem{ArmandoSTTT2013}
Solar-Lezama, A.: Program sketching. Int. J. Softw. Tools Technol. Transf.
  \textbf{15}(5-6),  475--495 (oct 2013). \doi{10.1007/s10009-012-0249-7},
  \url{https://doi.org/10.1007/s10009-012-0249-7}

\bibitem{Thistle2005}
Thistle, J.G.: Undecidability in decentralized supervision. Systems \& Control
  Letters  \textbf{54}(5),  503--509 (2005).
  \doi{10.1016/j.sysconle.2004.10.002}

\bibitem{TripakisIPL}
Tripakis, S.: {Undecidable Problems of Decentralized Observation and Control on
  Regular Languages}. Information Processing Letters  \textbf{90}(1),  21--28
  (Apr 2004). \doi{10.1016/j.ipl.2004.01.004}

\bibitem{DBLP:conf/pldi/UdupaRDMMA13}
Udupa, A., Raghavan, A., Deshmukh, J.V., Mador{-}Haim, S., Martin, M.M.K.,
  Alur, R.: {TRANSIT:} specifying protocols with concolic snippets. In: Boehm,
  H., Flanagan, C. (eds.) {ACM} {SIGPLAN} Conference on Programming Language
  Design and Implementation, {PLDI} '13, Seattle, WA, USA, June 16-19, 2013.
  pp. 287--296. {ACM} (2013). \doi{10.1145/2491956.2462174},
  \url{https://doi.org/10.1145/2491956.2462174}

\bibitem{YaoTGN22}
Yao, J., Tao, R., Gu, R., Nieh, J.: {DuoAI: Fast, Automated Inference of
  Inductive Invariants for Verifying Distributed Protocols}. In: Aguilera,
  M.K., Weatherspoon, H. (eds.) 16th USENIX Symposium on Operating Systems
  Design and Implementation (OSDI 2022). pp. 485--501. {USENIX} Association
  (2022), \url{https://www.usenix.org/conference/osdi22/presentation/yao}

\bibitem{YaoLivenessPOPL2024}
Yao, J., Tao, R., Gu, R., Nieh, J.: Mostly automated verification of liveness
  properties for distributed protocols with ranking functions. Proceedings of
  the ACM on Programming Languages (POPL)  \textbf{8},  1028--1059 (jan 2024).
  \doi{10.1145/3632877}, \url{https://doi.org/10.1145/3632877}

\bibitem{yao2021distai}
Yao, J., Tao, R., Gu, R., Nieh, J., Jana, S., Ryan, G.: {{DistAI}:
  {Data-Driven} Automated Invariant Learning for Distributed Protocols}. In:
  15th USENIX Symposium on Operating Systems Design and Implementation (OSDI
  2021). pp. 405--421. USENIX Association (Jul 2021),
  \url{https://www.usenix.org/conference/osdi21/presentation/yao}

\bibitem{tlcmodelchecker}
Yu, Y., Manolios, P., Lamport, L.: {Model Checking TLA+ Specifications}. In:
  Pierre, L., Kropf, T. (eds.) Correct Hardware Design and Verification
  Methods. pp. 54--66. Springer Berlin Heidelberg, Berlin, Heidelberg (1999)

\end{thebibliography}

\ifdefined\extendedversion
\appendix
\section{Exact Generalization of Stuttering Violations}
\label{sec:exact-stut}

For a stuttering counterexample, we can define the exact pruning constraint
$\pi_\textit{stut\_alt}(r,e,S)$ as follows:
\begin{align*}
    &\chi_\textit{unstut} := 
    \bigvee_{A\in\Acts_\textit{eq}} 
        \left(\bigwedge_{h\in\textit{pre}(A)} h(s)\neq\False\right)
        \wedge \left(\bigvee_{h\in\textit{post}(A)} h(s)\neq s[h.
        \textit{var}]\right)
    \\
    &\chi_\textit{deadlocked} :=
    \bigwedge_{A\in\Acts} \bigvee_{h\in\textit{pre}(A)} h(s)\neq\True
    \\
    &\pi_\textit{stut\_alt} := 
    \pi_\textit{safe}
    \vee\chi_\textit{enable}(\Acts_\textit{neq}, \{s_k\})
    \vee\chi_\textit{unstut}
    \vee\chi_\textit{deadlocked}
\end{align*}
where $\Acts$ is the set of all actions, $\Acts_\textit{eq}$ is the set of
fair, can-enable, {\em $s_k$-equivalent} actions, and $\Acts_\textit{neq}$ is
the set of fair, can-enable, actions that are not $s_k$-equivalent. If $x$ is a
non-hole post-clause, let $x.\textit{var}$ denote the state variable associated
with $x$ in the sketch $S$. An action is $s_k$-equivalent if for all non-hole
post-clauses $x$ in the action, $x(s_k) = s_k[x.\textit{var}]$; i.e., the
action does not change the value of the state variable associated with $x$ at
$s_k$, for all non-hole post-clauses $x$ appearing in the action.

Intuitively, the actions in $\Acts_\textit{neq}$ will never self-loop on $s_k$
because they change the value of some state variable at $s_k$; so it is
sufficient to merely enable one of those actions in order to avoid stuttering.
Otherwise, we must ensure that at least one of the actions in $\Acts_\textit{eq}$
is not only enabled, but also that at least one of the post-holes is filled with
an expression that changes the associated state variable at $s_k$.

$\pi_\textit{stut\_alt}$ differs from $\pi_\textit{stut}$. While
$\pi_\textit{stut}$ says that we can avoid stuttering by enabling a fair action
that is disabled, $\pi_\textit{stut\_alt}$ says that we must also ensure that
the newly enabled action does not self-loop. Likewise, $\pi_\textit{stut}$ says
that we can avoid stuttering by pointing an enabled fair action to a different
state, but $\pi_\textit{stut\_alt}$ additionally requires that this action
still be enabled. In other words, $\pi_\textit{stut}$ falls short of pruning
all completions that have $r$ because it does not require that the action that
would otherwise remedy the stuttering violation be both enabled and
non-self-looping.

\begin{lemma}
    \label{lem:pruning-stut-alt}
    Let $r$ be a stuttering counterexample of completion $e_1$ of sketch $S$.
    Let $\pi = \pi_\textit{stut\_alt}(r,e_1,S)$. Then for all completions $e_2$
    of $S$, $e_2$ satisfies $\pi$ if and only if $r$ is not a stuttering
    counterexample of $e_2$.
\end{lemma}
\begin{proof}

By construction, if $e_2$ satisfies $\pi$, then $e_2$ satisfies one of
$\pi_\textit{safe}(r,e_1,S)$, $\chi_\textit{move}(\Acts_\textit{enabled},
\{s_k\}, e_1)$, $\chi_\textit{enable}(\Acts_\textit{disabled}, \{s_k\})$, or
$\chi_\textit{deadlocked}$. In the first case, $r$ is not a counterexample of
$e_2$ because the path to $s_k$ is disabled. In the second case, $e_2$ enables
a fair action that is not $s_k$-equivalent, so the self-loop on $s_k$ is no
longer a fair cycle and the stuttering violation is avoided. In the third case,
$e_2$ enables a fair action that is $s_k$-equivalent, but also changes the
value of some state variable at $s_k$, so the self-loop on $s_k$ is no longer a
fair cycle and the stuttering violation is avoided. Finally, in the fourth
case, $s_k$ is a deadlock state in $e_2$, which technically avoids the
stuttering violation. Therefore, $r$ is not a stuttering counterexample of
$e_2$.

By construction, if $e_2$ does not satisfy $\pi$, then $e_2$ violates all of
$\pi_\textit{safe}(r,e_1,S)$, $\chi_\textit{move}(\Acts_\textit{enabled},
\{s_k\}, e_1)$, $\chi_\textit{enable}(\Acts_\textit{disabled}, \{s_k\})$, and
$\chi_\textit{deadlocked}$. In other words, all transitions in $r$ exist in
$e_2$, none of the non-$s_k$-equivalent fair actions are enabled at $s_k$, all
$s_k$-equivalent fair actions are either disabled or self-loop on $s_k$, and
$s_k$ is not a deadlock state in $e_2$. Therefore, $r$ is a stuttering
counterexample of $e_2$.
\qed

\end{proof}

\else
\fi

\end{document}